\documentclass[a4paper,10pt]{article}

\usepackage{pdfsync}
\usepackage[utf8]{inputenc}
\usepackage{mathtools}
\usepackage{amsthm}
\usepackage{amsmath}
\usepackage{amssymb}
\usepackage{mathabx}
\usepackage{tikz} 
\usepackage{hyperref}
\usepackage{bbold}
\usetikzlibrary{calc}
\usetikzlibrary{shapes.geometric}
\usetikzlibrary{decorations.pathreplacing}
\usetikzlibrary{intersections}
\usepackage[textwidth=2cm,textsize=small]{todonotes}
\usepackage{subcaption}
\usepackage{enumerate}
\usepackage{xspace}
\usepackage{xypic}
\usepackage{color}

\newtheorem{claim}{Claim}
\newtheorem{proposition}{Proposition}
\newtheorem{lemma}{Lemma}

\newtheorem*{remark}{Remark}

\newcommand{\Z}{\mathbb{Z}}
\newcommand{\N}{\mathbb{N}}
\newcommand{\Npar}[1]{\N_{\geq #1}}
\newcommand{\Pos}{\Npar{1}}
\newcommand{\states}{Q}
\newcommand{\step}[3]{#1 \stackrel{#2}{\leadsto} #3}
\newcommand{\Nrun}[2]{\xymatrix{#1 \ar[r] & #2}}
\newcommand{\NrunOp}[2]{\xymatrix{#1 & \ar[l]  #2}}
\newcommand{\Zrun}[2]{\xymatrix{#1 \ar@{.>}[r] & #2}}
\newcommand{\ZrunUog}[3]{\xymatrix{#1 \ar@{.>}[r]|-{#3} & #2}}
\newcommand{\ZrunUogOp}[3]{\xymatrix{#1 & \ar@{.>}[l]|-{#3} #2}}
\newcommand{\D}{\{1 \ldots d\}}
\newcommand{\size}[1]{\text{size}(#1)}

\newcommand{\problem}[3]{
\begin{quote}
\ \ {\sc #1:} \vspace{1mm}

\begin{tabular}{ll}
\emph{Input:} & #2\\
\emph{Question:} & #3
\end{tabular}
\end{quote}
}
\newcommand{\fold}[1]{\text{fold}(#1)}
\newcommand{\effect}[1]{\text{shift}(#1)}
\newcommand{\klejogol}[3]{#1 \oplus_{#3} #2}
\newcommand{\klej}[2]{\klejogol{#1}{#2}{}}

\newcommand{\klejt}[3]{\klej{#1}{\klej{#2}{#3}}}
\newcommand{\bad}[1]{{\color[rgb]{1,0,0} #1}}
\newcommand{\pr}[1]{{\color[rgb]{0,0,1}#1}}
\newcommand{\bd}{\bad{\bar {\small \Delta}}}
\newcommand{\bdp}{\bad{\bar {\small \Delta'}}}
\newcommand{\vass}{{\cal V}}
\newcommand{\gvass}{{\cal G}}

\title{VASS reachability  in three steps}
\author{S. Lasota} 

\begin{document}

\maketitle

This note is a product of digestion of the famous proof of decidability of the reachability problem for vector addition systems with states (VASS), as first established by Mayr~\cite{Mayr81,Mayr84} and then simplified by Kosaraju~\cite{Kosaraju82}.
The note is neither intended to be rigorously formal nor complete; it is rather intended to be an intuitive 
but precise enough description of main concepts exploited in the proof.
Very roughly, the overall idea is to provide a decidable condition $\Theta$ on a VASS such that $\Theta$ implies reachability
and $\neg \Theta$ implies that the size of VASS can be reduced.
With these two properties, the size of input can be incrementally reduced until the problem becomes trivial.
We proceed in three steps: we first formulate the condition $\Theta$ for plain VASS, then adapt it to more general
\emph{VASS with unconstrained coordinates}, and finally to \emph{generalized VASS} of~\cite{Kosaraju82}.

\section{The reachability problem}

A \emph{vector addition system with states} (VASS) consists of a finite set of control states $Q$ and a finite set $E \subseteq Q \times \Z^d \times Q$ of arcs.
The number $d\geq 1$ is the \emph{dimension} of a VASS.
A pseudo-configuration is a pair $(q, v) \in \states \times \Z^d$; it is a \emph{configuration} if $v\in\N^d$.
An arc $e = (q, z, q')$ induces a step
\[
\step{(q, v)}{e}{(q', v+z)}
\]
between pseudo-configurations.
We write $\Zrun{q, v}{q', v'}$ if there is a sequence of steps from $(q, v)$ to $(q', v')$; every such sequence we call \emph{pseudo-run}. 
We reserve this term for a sequence of steps, as well as for an (inducing) sequence of arcs.
If all vectors appearing in a pseudo-run belong to $\N^d$ we call it \emph{run}, and write $\Nrun{q, v}{q', v'}$;
this implies in particular that $(q, v)$ and $(q', v')$ themselves are configurations.  
The aim of this note is to describe an algorithm for
\problem{VASS reachability problem}
{a VASS $(d, Q, E)$ and two configurations $(q, v), (q', v')$.}
{does $\Nrun{q, v}{q', v'}$ hold?}

\paragraph*{Sufficient condition}

As a warm-up, we prove a sufficient condition for reachability. For a VASS and two configurations 
$(q, v)$, $(q', v')$, define the following two conditions:
\begin{enumerate}
\item[$\Theta_1$:] For every $m \geq 1$, $\Zrun{q, v}{q', v'}$ by a pseudo-run that uses every arc at least $m$ times.
\item[$\Theta_2$:] There are vectors $\Delta, \Delta' \geq \vec 1$
such that
\begin{align*}
& \ \Nrun{q, v}{q, v +\Delta} \\
& \NrunOp{q', v'}{q', v' + \Delta'}
\end{align*}
\end{enumerate}
\begin{proposition}
$\Theta_1 \land \Theta_2$ implies $\Nrun{q, v}{q', v'}$.
\end{proposition}
\begin{proof}
We will use the following claim, to be proved later:
\begin{claim} \label{c:Deltas}
$\Zrun{q', \Delta}{q', \Delta'}$.
\end{claim}
\noindent
Here is a shape of a required run from $(q, v)$ to $(q', v')$:
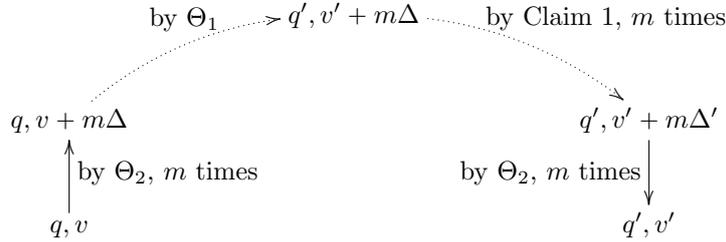
\[
\xymatrix{
&& q', v' + m \Delta \ar@/^1pc/@{.>}[rrd]^{\txt{\ \ \ \ \ \ \ \ by Claim~\ref{c:Deltas}, $m$ times}} \\
q, v + m \Delta  \ar@/^1pc/@{.>}[rru]^{\txt{by $\Theta_1$}} &&&& q', v' + m \Delta' \ar[d]_{\txt{by $\Theta_2$, $m$ times}} \\
q, v \ar[u]_{\txt{by $\Theta_2$, $m$ times}} &&&& q', v'
}
\]
Observe that when $m$ increases, the three intermediate points also increase on all coordinates. 
Therefore, for a sufficiently large $m$, the two pseudo-runs become runs. 
\end{proof}

\paragraph*{Proof of Claim~\ref{c:Deltas}.}

Consider the underlying graph of the VASS, whose vertices are control states
and edges are arcs (note that there may be parallel edges).
Every pseudo-run induces a path in the graph.
For a pseudo-run from $(p, w)$ to $(p', w')$, we shortly speak of a pseudo-run from $p$ to $p'$ when vectors $w, w'$ are irrelevant.
Let $E$ denote the set of arcs. 
By the \emph{folding} of a pseudorun $\pi$ we mean the vector $\fold{\pi} \in \N^E$ that says how many times every arc is used by $\pi$.
The following lemma, roughly speaking, allows us to subtract one pseudo-run from another
(it is proved using Eulerian equalities):
\begin{lemma} \label{l:odejmowanie}
Let $\tau, \rho$ be two pseudo-runs from $p$ to $p'$ such that\footnote{
We write $\vec m_C$ for a constant vector in $\Z^C$ having $m$ on all coordinates.
We prefer to omit the subscript $C$ and write simply $\vec m$ whenever this does not lead to confusion.
}
$$\fold{\tau} - \fold{\rho} \geq \vec 1_E.$$
For every non-isolated control state $p''$ there is a pseudo-run $\sigma$ from $p''$ to $p''$ with
$\fold{\sigma} = \fold{\tau} - \fold{\rho}$.
\end{lemma}
By $\effect{\pi} \in \Z^d$ we mean the effect of a pseudo-run $\pi$, namely the difference between its final vector and its initial one. 
Note that the shift of a pseudo-run is completely determined by its folding.
To prove Claim~\ref{c:Deltas}, we need to show that there is a pseudo-run from $q'$ to $q'$ with shift $\Delta' - \Delta$.

Basing on condition $\Theta_1$, we know that we can pick  two pseudo-runs $\tau$, $\rho$ from $(q, v)$ to $(q', v')$ with arbitrarily large
difference  $\fold{\tau} - \fold{\rho}$.
Fix (due to $\Theta_2$) a run $\pi$ from $(q, v)$ to $(q, v + \Delta)$, and a run $\pi'$ from $(q', v' + \Delta')$ to $(q', v')$.
Then fix two  pseudoruns $\tau$, $\rho$ from $(q, v)$ to $(q', v')$ such that
\[
\fold{\tau} - \fold{\rho} - \fold{\pi} - \fold{\pi'} \geq \vec 1.
\]
Finally, apply Lemma~\ref{l:odejmowanie} three times in a sequence, to deduce that there is a pseudo-run $\nu$ from $q'$ to $q'$ satisfying
\[
\fold{\nu} = \fold{\tau} - \fold{\rho} - \fold{\pi} - \fold{\pi'}.
\]
Indeed,  $\effect{\nu} = \effect{\tau} - \effect{\rho} - \effect{\pi} - \effect{\pi'} = \Delta' - \Delta$ as required.

\section{Partially unconstrained reachability problem} \label{sec:VASSunc}

We now slightly generalize the reachability problem, and a sufficient condition.
In the next section we will provide a yet further generalization that will be finally suitable for designing a decision procedure for reachability.

We will need a bit of concise notation. 
From now on we identify $\Z^d$ and $\Z^{\D}$; for instance, the set of configurations is $Q\times \N^{\D}$. 
For two disjoint subsets $C, B \in \D$ and two vectors $v \in \Z^C$ and $w \in \Z^{B}$, we write
$
\klejogol{v}{w}{}
$ 
for the unique vector in $\Z^{C\cup B}$ obtained by glueing together $v$ and $w$.
Formally:
\[
(\klejogol{v}{w}{})(i) = \begin{cases}
v(i) & \text{ if } i \in C \\
w(i) & \text{ if } i \in B.
\end{cases}
\]
%
From now on, by convention $\bar C$ will always  denote the complement $\D - C$. 

The generalization of the reachability problem amounts to considering only some subset $C \subseteq D$ of coordinates
as \emph{constrained}, while the remaining coordinates (i.e., those in $\bar C$) are considered \emph{unconstrained}. The input and output configuration
is specified only on constrained coordinates, and left unspecified on the remaining ones. Nevertheless, a run we ask for should remain
nonnegative on all coordinates. Here is a precise formulation: 

\problem{Partially unconstrained VASS reachability problem}
{a VASS $(d, Q, E)$, two subsets $C, C' \subseteq \D$, \\ 
&  $(q, v) \in \states \times \N^C$ and $(q', v') \in \states \times \N^{C'}$.}
{does $\Nrun{q, \klej{v}{\bar v}}{q', \klej{v'}{\bar v'}}$ hold \\
& for some vectors $\bar v \in \N^{\bar C}$, $\bar v' \in \N^{\bar C'}$?}
%
We remark that we do not assume $C = C'$. The setting of the previous section is the special case $C = C' = \D$.

\paragraph*{Sufficient condition}

Here is a generalization of $\Theta_1$ and $\Theta_2$ to the more general setting.
We write $\ZrunUog{q, v}{q', v'}{C}$, for $C\subseteq \{1 \ldots d\}$,
to say that there is a pseudo-run from $(q, v)$ to $(q', v')$ whose all vectors are non-negative on coordinates from $C$ (such pseudo-runs we call \emph{$C$-runs}).
We write shortly $\Npar{m}$ for $\N - \{0 \ldots m-1\}$.
\begin{enumerate}
\item[$\Theta_1$:] For every $m \geq 1$, there are some vectors $\bar v \in (\Npar{m})^{\bar C}$, $\bar v' \in (\Npar{m})^{\bar C'}$ such that
 $\Zrun{q, \klej{v}{\bar v}}{q', \klej{v'}{\bar v'}}$ by a pseudo-run that traverses every arc at least $m$ times.
\item[$\Theta_2$:] There are vectors $\Delta \in (\Pos)^C$, $\Delta' \in (\Pos)^{C'}$, $\bd \in \Z^{\bar C}$ and $\bdp 
\in \Z^{\bar C'}$ 
such that
\begin{align*}
\pi: \quad & \ \ZrunUog{q, \klej{v}{\vec 0}}{q, \klej{(v +\Delta)}{\bd}}{C} \\
\pi': \quad & \ZrunUogOp{q', \klej{v'}{\vec 0}}{q', \klej{(v' + \Delta')}{\bdp}}{C'}
\end{align*}
\end{enumerate}
\begin{proposition}
$\Theta_1 \land \Theta_2$ implies $\Nrun{q, \klej{v}{\bar v}}{q', \klej{v'}{\bar v'}}$ for some vectors $\bar v \in \N^{\bar C}$, $\bar v' \in \N^{\bar C'}$.
\end{proposition}
\begin{proof}
The general idea of the proof is similar to the previous section, namely pumping up by a multiplicity of $\Delta$ 
(and de-pumping down by the same multiplicity of $\Delta'$) in order to make
some pseudorun  $\Zrun{q, \klej{v}{\bar v}}{q', \klej{v'}{\bar v'}}$ into a run.
The new difficulty is that pumping involves $\klej{\Delta}{\bd}$, with $\bd$ possibly negative on some coordinates (and likewise for de-pumping).
This issue is solved by starting from $\klej{v}{\bar v}$, for a sufficiently large $\bar v \geq \vec m$. 

%

We will need a couple of facts.
The first one easily follows from $\Theta_1$:
\begin{claim} \label{c:linear}
There are vectors $\bar v \in \N^{\bar C}$,
$\bar v' \in \N^{\bar C'}$ 
and a pseudo-run
\begin{align*} 
\pi_0: \quad & \Zrun{q, \klej{v}{\bar v}}{q', \klej{v'}{\bar v'}} 
\end{align*}
such that for every $m>0$ there is a pseudo-run
%
\begin{align*} 
\pi_1: \quad & \Zrun{q, \klej{v}{(\bar v + \bar \delta)}}{q', \klej{v'}{(\bar v' + \bar \delta')}}
\end{align*}
with $\fold{\pi_1} - \fold{\pi_0} \geq \vec m_E$ and $\bar \delta \geq \vec m_{\bar C}$ and $\bar \delta' \geq \vec m_{\bar C'}$.
\end{claim}
\noindent
In other words, $\pi_0$ and $\pi_1$ can be chosen to make the three vectors $\fold{\pi_1} - \fold{\pi_0}$, \ $\bar \delta$ and $\bar \delta'$
arbitrarily large on all coordinates. Therefore we conclude:
\begin{claim} \label{c:Deltastwoaux}
The pseudo-runs $\pi_0$ and $\pi_1$ can be chosen so that:
\begin{itemize}
\item[(a)\,] \ $\bar \delta + \bd \geq \vec 1_{\bar C}$, \  
$\bar \delta' + \bdp \geq \vec 1_{\bar C'}$ 
\vspace{-1mm}
\item[(b)\,] pseudo-runs in $\Theta_2$, lifted by $\klej{\vec 0}{(\bar v + \bar \delta)}$ and 
$\klej{\vec 0}{(\bar v' + \bar \delta')}$, respectively, become runs:
\begin{align*}
\pi: \quad & \ \Nrun{q, \klej{v}{(\bar v + \bar \delta)}}{q, \klej{(v + \Delta)}{(\bar v + \bar \delta + \bd)}} \\
\pi': \quad & \NrunOp{q', \klej{v'}{(\bar v' + \bar \delta')}}{q', \klej{(v' + \Delta')}{(\bar v' + \bar \delta' + \bdp)}}
\end{align*}
\item[(c)\,] $\fold{\pi_1} - \fold{\pi_0} - \fold{\pi} - \fold{\pi'} \geq \vec 1_E$.
\end{itemize}
\end{claim}
\noindent
Using Claim~\ref{c:Deltastwoaux}(a)--(b) together with the monotonicity of VASSes,
we deduce that for an arbitrary $m>0$, the runs $\pi$ and $\pi'$ can be repeated $m$ times when
lifted further by $\klej{\vec 0}{m \bar \delta}$ and $\klej{\vec 0}{m \bar \delta'}$, respectively:
\begin{claim} \label{c:end}
For 
every $m \geq 1$ it holds
\begin{align*}
& \ \ \Nrun{q, \klej{v}{(\bar v + m \bar \delta)}}{q, \klej{(v + m \Delta)}{(\bar v + m (\bar \delta + \bd))}} \\
& \NrunOp{q', \klej{v'}{(\bar v' + m \bar \delta')}}{q', \klej{(v' + m \Delta')}{(\bar v' + m (\bar \delta' + \bdp))}}
\end{align*}
\end{claim}
\noindent
The last claim generalizes Claim~\ref{c:Deltas} from the previous section.
\begin{claim} \label{c:Deltastwo}
$\Zrun{q', \klej{\Delta}{(\bar \delta + \bd)}}{q', \klej{\Delta'}{(\bar \delta' + \bdp)}}$.
\end{claim}
\begin{proof} 
We have $\effect{\pi_1} - \effect{\pi_0} = (\klej{\vec 0}{\bar \delta'}) - (\klej{\vec 0}{\bar \delta})$ and
$\effect{\pi} = \klej{\Delta}{\bd}$ and $\effect{\pi'} = \klej{(-\Delta')}{(-\bdp)}$.
By Claim~\ref{c:Deltastwoaux}(c) we can apply Lemma~\ref{l:odejmowanie} 
three times, to deduce that there is a pseudo-run $\nu$ from $q'$ to $q'$ satisfying
\[
\fold{\nu} = \fold{\pi_1} - \fold{\pi_0} - \fold{\pi} - \fold{\pi'}.
\]
We check:  $\effect{\nu} = (\effect{\pi_1} - \effect{\pi_0}) - \effect{\pi} - \effect{\pi'} = 
\klej{\Delta'}{(\bar \delta' + \bdp)} -
\klej{\Delta}{(\bar \delta + \bd)}$, as required.
\end{proof}
We are now prepared to draw a shape of a required run (for readability, the primed items are depicted in blue):
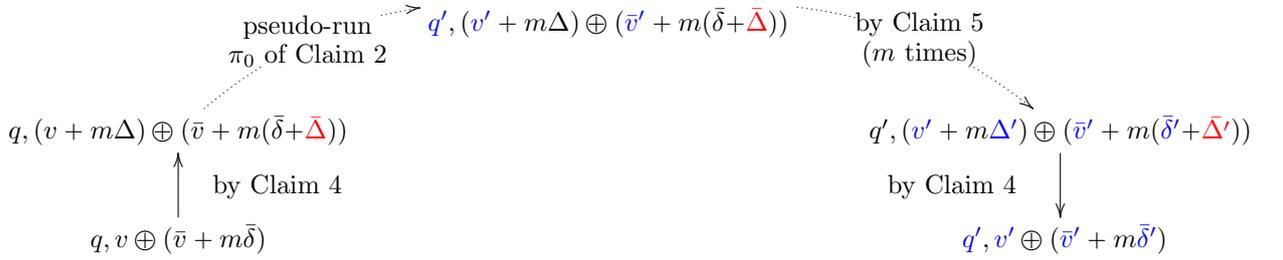
\[ \hspace{-1cm}
\xymatrix{
&  
      \pr{q'}, \klej{(\pr{v'} + m \Delta)}{(\pr{\bar v'} + m (\bar \delta {+} \bd))}
      \ar@/^2pc/@{.>}[rd]|-{\txt{by Claim~\ref{c:Deltastwo}\\ ($m$ times)}} \\
 \klej{q, (v + m \Delta)}{(\bar v + m (\bar \delta {+} \bd))}  \ar@/^2pc/@{.>}[ru]|-{\txt{pseudo-run\\ $\pi_0$ of Claim~\ref{c:linear}}} &&
\klej{q', (\pr{v'} + m \pr{\Delta'})}{(\pr{\bar v'} + m (\pr{\bar \delta'} {+} \pr{\bdp}))}  \ar@{->}[d]_{\txt{by Claim~\ref{c:end} \ \ \ \ }} \\
\klej{q, v}{(\bar v + m \bar \delta)} \ar@{->}[u]_{\txt{\ \ \ by Claim~\ref{c:end}}} && \ \klej{\pr{q'}, \pr{v'}}{(\pr{\bar v'} + m \pr{\bar \delta'})}
}
\]
When $m$ increases, each of the three intermediate points increases on all coordinates. 
As a conclusion, for sufficiently large $m$, all the pseudo-runs become runs.
\end{proof}

\begin{remark} \rm
For the next section it is important to note that we have actually shown 
$\Nrun{q, \klej{v}{(\bar v + m \bar \delta)}}{q', \klej{v'}{(\bar v' + m \bar \delta')}}$
for all sufficiently large $m$.
\end{remark}


\section{Generalized reachability problem}

We do now the last generalization in order to complete the decidability proof. 
By a \emph{component} we mean a VASS $(d, Q, E)$ together with the following data:
\begin{itemize}
\item initial and final state $q, q' \in Q$;
\vspace{-1.5mm}
\item subset of rigid coordinates $R \subseteq \D$; we assume that all arcs in $E$ have $0$ on all coordinates in $R$
and hence, intuitively speaking, $d - |R|$ may be considered as the actual dimension of the component;
\vspace{-1.5mm}
\item rigid vector $r \in \N^{R}$;
\vspace{-1.5mm}
\item two partitions $\D - R = C \cup U = C' \cup U'$ of non-rigid coordinates into initial constrained coordinates $C$ and initial unconstrained coordinates $U$, and into final constrained coordinates $C'$ and final unconstrained coordinates $U'$;
\vspace{-1.5mm}
\item initial and final vector $v \in \N^{C}$, $v' \in \N^{C'}$.
\end{itemize}
Note that component does not essentially differ from the input of the partially unconstrained reachability problem from the previous section.
The \emph{generalized VASS} (GVASS) $\gvass$ consists of $l\geq 1$ components
\[
\vass_i = (d, E_i, Q_i, q_i, q'_i, R_i, r_i, C_i, U_i, C'_i, U'_i, v_i, v'_i)
\]
of the same dimension $d$, with pairwise disjoint state sets $Q_i$, and
$l{-}1$ arcs of the form $e_i = (q'_i, z_i, q_{i+1})$, where $z_i \in \Z^d$, for $i \in \{1 \ldots i-1\}$.
%
%
We will be interested in pseudo-runs $\pi$ from $q_1$ to $q'_l$ of the following form:
\begin{align} \label{eq:runPU}
\begin{aligned}
& \Zrun{q_1, \klejt{r_1}{v_1}{u_1}}{q'_1, \klejt{r_1}{v'_1}{u'_1}  \step{}{e_1}{}} \\
& \Zrun{q_2, \klejt{r_2}{v_2}{u_2}}{q'_2, \klejt{r_2}{v'_2}{u'_2}  \step{}{e_2}{}} \ \ldots \\
 \ldots  \step{}{e_{l-1}}{}  & \Zrun{\ q_l, \klejt{r_l}{v_l}{u_l}}{\ q'_l, \klejt{r_l}{v'_l}{u'_l}}
\end{aligned}
\end{align}
for some $u_1\in\N^{U_1}, u'_1\in \N^{U'_1}, \ldots, u_l \in \N^{U_l}, u'_l\in \N^{U'_l}$.
Each such pseudo-run $\pi$ passes through every arc $e_i$ exactly once, 
and thus splits into $l$ pseudo-runs 
$
\pi = \pi_1 e_1 \pi_2 e_2 \ldots e_{l-1} \pi_l,
$ 
each $\pi_i$ being a pseudo-run in $\vass_i$. When each of $\pi_i$ is a run, 
we call $\pi$ a run; 
if such a run exists we say that $\gvass$ \emph{admits reachability}.

\problem{Generalized VASS reachability problem}
{a GVASS $\gvass$.}
{does $\gvass$ admit reachability?}
%
The setting of the previous section is the special case of one component without rigid coordinates: $l = 1$, $R_1 = \emptyset$.

\subsubsection*{Sufficient condition}

The condition $\Theta_2$ below is essentially the conjunction of conditions $\Theta_2$ of the previous section for each of the VASSes 
$\vass_i$ separately; the only difference is taking rigid coordinates into account.
On the other hand, the condition $\Theta_1$ below speaks jointly about all the VASSes $\vass_i$. 

\begin{enumerate}
\item[$\Theta_1$:] For every $m \geq 1$, there is a pseudo-run  from $q_1$ to $q'_l$ of the form~\eqref{eq:runPU} 
that traverses every arc in every $E_i$ at least $m$ times, for some 
$u_1 \in (\Npar{m})^{U_1}, u'_1 \in (\Npar{m})^{U'_1}, \ldots, u_l \in (\Npar{m})^{U_l}, u'_l \in (\Npar{m})^{U'_l}$.
%
\item[$\Theta_2$:] For every $i \in \{1\ldots l\}$ there are vectors $\Delta \in (\Pos)^{C_i}$, $\Delta' \in (\Pos)^{C'_i}$, 
$\bd \in \Z^{U_i}$ and $\bdp \in \Z^{U'_i}$ 
such that
\begin{align}
& \ZrunUog{q_i, \klejt{r_i}{v_i}{\vec 0}}{q_i, \klejt{r_i}{(v_i + \Delta)}{\bd}}{C_i} \label{eq:theta2} \\
& \ZrunUogOp{q'_i, \klejt{r_i}{v'_i}{\vec 0}}{q'_i, \klejt{r_i}{(v'_i + \Delta')}{\bdp}}{C'_i} \label{eq:theta2bis}
\end{align}
\end{enumerate}
Observe that $\Theta_1$ implies $C'_i = C_{i+1}$ for $i \in \{1\ldots l{-}1\}$.
The sufficient condition for reachability is proved similarly as in the previous section:
\begin{proposition}
If $\gvass$ satisfies $\Theta_1 \land \Theta_2$ then $\gvass$ admits reachability.
\end{proposition}
\noindent
Indeed, in Claim~\ref{c:linear} one should consider all components simultaneously and recall the remark at the end of
Section~\ref{sec:VASSunc};
for the other claims and the construction of a run, one can consider the components separately.

Furthermore, the sufficient condition can be effectively tested:
\begin{proposition}
Both $\Theta_1$ and $\Theta_2$ are decidable.
\end{proposition}
\begin{proof}
Pseudo-runs~\eqref{eq:runPU} can be encoded as the set of nonegative solutions of a system of linear Diophantine equations.
Then condition $\Theta_1$ can be decided by inspecting the (hybrid-linear) set of solutions (cf.~Claim~\ref{c:hybrid} below).
Checking condition $\Theta_2$ reduces to the coverability problem.
\end{proof}

\subsubsection*{Refinement}
Let  $\size{\vass_i} = (d - |R_i|, |E_i|, |U_i|+|U'_i|) \in \N^3$. Thus the size of $\vass_i$ is a triple consisting of: 
the number of non-rigid coordinates, the number of arcs, the number of unconstrained coordinates. 
For a GVASS $\gvass$, we define $\size{\gvass}$ as the multiset of sizes of all components $\vass_i$.

Order triples in $\N^3$ lexicographically.
For two finite multisets  of triples $m$ and $m'$, we say that $m'$ \emph{refines} $m$ if $m'$ is obtained by removing one triple from $m$, 
and replacing it by a finite number of lexicographically strictly smaller triples. 

\begin{claim}
The refinement relation is well-founded.
\end{claim}

We shortly say that $\gvass'$ refines $\gvass$ when $\size{\gvass'}$ refines $\size{\gvass}$.
We can assume wlog.~that every component of $\gvass$ is strongly connected:
\begin{claim}
If the underlying graph\footnote{We ignore here one inessential detail: this is a \emph{multigraph}, i.e., parallel edges are allowed.} of some component of $\gvass$ is not strongly-connected then one can compute 
$\gvass_1 \ldots \gvass_n$ refining $\gvass$ such that $\gvass$ admits reachability if, and only if some of $\gvass_1 \ldots \gvass_n$ does.
\end{claim}
\noindent
Indeed, it suffices to do the decomposition into strongly connected graphs.

For \emph{trivial} $\gvass$, whose size contains only zero triples $(0, 0, 0)$, the reachability problems trivializes. 
Otherwise, either $\gvass$ satisfies $\Theta_1 \land \Theta_2$ and thus admits reachability, or $\gvass$ can be refined:
\begin{proposition}
If a non-trivial $\gvass$ violates $\Theta_1$ then one can compute 
$\gvass_1 \ldots \gvass_n$ refining $\gvass$ such that $\gvass$ admits reachability if, and only if some of $\gvass_1 \ldots \gvass_n$ does.
\end{proposition}
\begin{proof}
Wlog.~assume that the underlying graphs of all components $\vass_i$ are strongly connected.
Let 
$
k = \sum_{i = 1}^l |U_i| + E_i + |U'_i|.
$
Consider the set $L \subseteq \N^k$ of all vectors
\begin{align}  \label{e:fold} 
(u_1, f_1, u'_1, \ldots, u_l, f_l, u'_l) \in \N^k
\end{align}
such that there is a pseudo-run $\pi = \pi_1 e_1 \pi_2 e_2 \ldots e_{l-1} \pi_l$ 
of the form~\eqref{eq:runPU} with $\fold{\pi_1} = f_1 \geq \vec 1, \ldots, \fold{\pi_l} = f_l \geq \vec 1$.
The set $L$ is the set of nonnegative solutions of a system of linear Diophantine equations, and thus we have:
\begin{claim} \label{c:hybrid}
One can compute finite sets $B, P \subseteq \N^k$ such that $L = B + P^*$.
\end{claim}

Suppose $\gvass$ does not satisfy $\Theta_1$. Hence for some coordinate  in $\{1 \dots k\}$, 
all vectors in $P$ have zero on that coordinate.
This \emph{zero coordinate} corresponds either to some arc, or to some unconstrained (input or output) coordinate.

Suppose the first case holds, and let $e\in E_i$ be the arc corresponding to the zero coordinate.
By Claim~\ref{c:hybrid} one can compute a number $c$ such that every pseudo-run~\eqref{eq:runPU} passes through $e$ at most $c$ times.
We refine $\gvass$ by $c+1$ GVASSes $\gvass_0 \ldots \gvass_c$, each $\gvass_m$ obtained by replacing $\vass_i$ by a sequence of $m+1$ copies of $\vass_i - \{e\}$, 
i.e.~of $\vass_i$ without the arc $e$.
The rigid coordinates and rigid vector of all copies are as in $\vass_i$.
The initial constrained coordinates of the first copy are $C_i$, the final constrained coordinates of the last copy are $C'_i$, and
the remaining initial or final constrained coordinates of all copies are empty sets. 
The initial vector of the first copy is $v_1$ and the final vector of the last copy is $v'_l$; all other initial and final vectors are empty ones. 

Now suppose the second case holds, i.e., the zero coordinate corresponds to some, say,
initial unconstrained coordinate $j \in U_i$ (final unconstrained coordinate is treated symmetrically).
By Claim~\ref{c:hybrid} one can compute a number $c$ such that the value on coordinate $j$ in $u_i$ 
(cf.~\eqref{e:fold})
is at most $c$, for every pseudo-run $\pi$.
We refine $\gvass$ by constraining the coordinate $j$ to some value in $\{0 \ldots c\}$.
We define $c+1$ refining GVASSes $\gvass_0 \ldots \gvass_c$, where $\gvass_m$ differs from $\gvass$ only
by making the coordinate $j$ in $\vass_i$ an initial constrained coordinate, with value $m$.
\end{proof}

\begin{proposition}
If a non-trivial $\gvass$ 
violates $\Theta_2$ then one can compute 
$\gvass_1 \ldots \gvass_n$ refining $\gvass$ such that $\gvass$ admits reachability if, and only if some of $\gvass_1 \ldots \gvass_n$ does.
\end{proposition}
\begin{proof}
Wlog.~assume that the underlying graphs of all components $\vass_i$ are strongly connected.
Suppose that $\gvass$ does 
not satisfy $\Theta_2$, i.e.~condition~\eqref{eq:theta2} fails for some $i$ 
(condition~\eqref{eq:theta2bis} is treated symmetrically).
Thus all initial constrained coordinates can not be simultaneously increased arbitrarily, which means that
for some number $c$, in every pseudo-configuration reachable in $\vass_i$ from
$\klej{v_i}{\vec 0}$ via the relation $\ZrunUog{}{}{C_i}$, some of initial constrained coordinates $j \in C_i$ is bounded by $c$.
From the coverability tree for $\vass_i$ one can extract $c$ with a stronger property:
\begin{claim}
For every $C_i$-run $\pi$ in $\vass_i$ from $\klej{v_i}{\vec 0}$ there is
an initial constrained coordinate $j \in C_i$ which is bounded by $c$ in $\pi$.
\end{claim}
\noindent
Relying on the claim, we refine $\gvass$ by a finite family of GVASSes.
For every $j \in C_i \cap C'_i$ the family contains one GVASS $\gvass_j$, and for every $j \in C_i \cap U'_i$ the family contains
$c+1$ GVASSes $\gvass_{j, 0} \ldots \gvass_{j, c}$, as outlined below:

\paragraph{$j \in C_i \cap U'_i$:} Thus $j$ is  a final unconstrained coordinate.
We define GVASSes $\gvass_{j,0} \ldots \gvass_{j,c}$, where $\gvass_{j, m}$ differes from $\gvass$ only by
making the coordinate $j$ a final constrained coordinate in $\vass_i$, and fixing its value to $m$.

\paragraph{$j \in C_i \cap C'_i$:} Thus $j$ is a final constrained coordinate. Let $a$ and $a'$ be the values of initial and final vectors
$v_i$, $v'_i$ on coordinate $j$.
We define $\gvass_j$ by replacing $\vass_i$ with two components $\vass'$ and $\vass''$. 
$\vass'$ behaves exactly as $\vass_i$ with the only exception that the value of the $j$th coordinate is kept between $0$ and $c$.
This can be achieved using a cross-product of $\vass_i$ with a finite state automaton, with states $\{0, \ldots, c\}$, the initial state $a$,
the final state $a'$, and transitions induced by the $j$th coordinate of arcs in $E_i$.
This allows to set the $j$th coordinate of all arcs in $\vass'$ to $0$; in consequence, the coordinate $j$ can be moved to rigid coordinates 
of $\vass'$.
Thus $\vass'$ has $(c+1)$ times more states and arcs than $\vass_i$ but one less non-rigid coordinate.
The rigid vector of $\vass'$ is set to $a$ on coordinate $j$.  
The difference $a' - a$ is easily compensated by adding one arc-less component $\vass''$ to $\gvass_j$, connected to $\vass'$ by an arc
that adds $a' - a$ on coordinate $j$ and preserves all other coordinates.
\end{proof}

\bibliographystyle{plain}
\bibliography{bib}
\end{document}